\pdfoutput=1
\documentclass[imslayout,noinfoline,a4paper,12pt,twoside]{imsart}
\RequirePackage[OT1]{fontenc} 
\RequirePackage{amsthm,amsmath,amssymb,natbib} 

\startlocaldefs
\theoremstyle{plain}
\newtheorem{thm}{Theorem}[section]
\newtheorem{prop}{Proposition}[section]
\newtheorem{lem}{Lemma}[section]
\endlocaldefs
\let\oldsqrt\sqrt
\def\sqrt{\mathpalette\DHLhksqrt}
\def\DHLhksqrt#1#2{%
\setbox0=\hbox{$#1\oldsqrt{#2\,}$}\dimen0=\ht0
\advance\dimen0-0.2\ht0
\setbox2=\hbox{\vrule height\ht0 depth -\dimen0}%
{\box0\lower0.4pt\box2}}
\renewcommand{\[}{\begin{eqnarray*}}
\renewcommand{\]}{\end{eqnarray*}}
\newcommand{\la}{\begin{eqnarray}}
\newcommand{\al}{\end{eqnarray}}
\DeclareMathOperator*{\bigtimes}{\mbox{\LARGE$\times$}}

\renewcommand{\epsilon}{\varepsilon}
\renewcommand{\phi}{\varphi}

\newcommand{\N}{{\mathbb N}}
\newcommand{\R}{{\mathbb R}}

\newcommand{\cA}{{\mathcal A}}\newcommand{\cB}{{\mathcal B}}

\newcommand{\cP}{{\mathcal P}}\newcommand{\cQ}{{\mathcal Q}}
\newcommand{\cX}{{\mathcal X}}
\newcommand{\cY}{{\mathcal Y}}

\newcommand{\Prob}{\mbox{\rm Prob}}
\newcommand{\Mark}{\mbox{\rm Mark}}
\newcommand{\im}{\mbox{\,\tiny$\square$}\,} 
\newcommand{\dd}{{\,\text{\rm d}}}

\newcommand{\1}{\bold 1}
\renewcommand{\epsilon}{\varepsilon}\renewcommand{\phi}{\varphi}%
\renewcommand{\rho}{\varrho}\renewcommand{\theta}{\vartheta}    %
\newcommand{\cxle}{\le_{\text{\rm cx}}} 

\newcommand{\Halmos}{$\quad\square$}

\usepackage{Sweave}
\begin{document}

\begin{frontmatter}


\title{Combining individually valid and conditionally i.i.d. P-variables}

\runtitle{Combining P-variables}

\begin{aug}
\author{\fnms{Lutz} \snm{Mattner}
\ead[label=e1]{mattner@uni-trier.de}}

\runauthor{Lutz Mattner}

\affiliation{Universit\"at Trier} 
 {\rm \today}\\
 {\footnotesize\tt \jobname.tex}

\address{Universit\"at Trier\\ Fachbereich IV -  Mathematik\\ 
54286 Trier\\ Germany\\ \printead{e1}}


\end{aug}

\begin{abstract}
For a given testing problem, let $U_1,\ldots,U_n$
be individually valid and conditionally on the data
i.i.d.\ P-variables (often called P-values).
For example, the data could come in groups,
and each $U_i$ could be based on subsampling 
just one datum from each group, 
in order to justify an independence assumption under the hypothesis.
The problem is then to deterministically combine the $U_i$ 
into a valid summary P-variable. 
Restricting here our attention to  functions of
a given order statistic $U_{k:n}$ of the $U_i$, 
we compute the function $f_{n,k}$ which is smallest
among all increasing functions $f$ such
that  $f(U_{k:n})$ is always a valid P-variable under
the stated assumptions. Since $f_{n,k}(u)\le 1\wedge \left(\frac {n}{k} u\right)$,
with the right hand side being a good approximation 
for the  left when $k$ is large, one may in particular always take the 
minimum of $1$ and twice the left sample median of the given P-variables.  

We sketch the original application of the above 
in a recent study of associations between 
various primate species by Astaras et al.

\end{abstract}


\end{frontmatter}



\section{Introduction and main result} 
\subsection{The problem and a summary of its tentative solution}
\label{The.problem.etc}
This paper is motivated by a practical example of  the following kind:
Suppose that, for testing a hypothesis of interest, we have many 
observations collected over time and a test, or more precisely a family 
of tests indexed by the observed sample size, which is judged to  be 
reasonable if independence of the observations could be assumed. If this latter
assumption is in doubt only for observations close
in time, say those taken on the same day, one might
consider choosing  one per day   at random in order to 
apply the test only to those chosen. While this would be   
theoretically correct, or at least less faulty
under the null hypothesis than ignoring the dependence 
issue, its reported  result could be challenged by asking whether perhaps, 
after several trials, specifically a subset of the 
observations leading to the desired rejection had been chosen. 
Hence, given the data, one might consider
repeating the random choice independently $n$
times. If the  test statistics are written as  P-variables,
see Subsection~\ref{subsec.stats.notation} below for formal definitions,
this yields an example of the following more general problem: 
Given $n$ individually valid and conditionally on the data
i.i.d.~P-variables $U_1,\ldots,U_n$, find 
an appropriate summary P-variable, that is, find a 
permutation invariant measurable function
$F:[0,1]^n\rightarrow[0,1]$ such that $V:= F(U_1,\ldots,U_n)$
is a reasonable P-variable for our testing problem. 

The  main result of this paper, Theorem~\ref{Main.result} below, gives  
simple solutions to this general problem, with  each optimal 
within the set of all summary P-variables being isotone functions of 
a fixed order statistic $U_{k:n}$. 

For practical purposes the results of this paper may be summarized 
as follows:  If $k\in \{1,\ldots,n\}$ is fixed in advance, then 
$V:=1\wedge (\frac{n}{k} U_{k:n})$ is a  P-variable under the stated 
assumptions. For $U_{k:n}$ small, and this is the only case of interest,
and for $k$ rather large, this $V$  is approximately the smallest isotone 
function of  $U_{k:n}$ being a P-variable. For moderate $k$, the constant
$\frac nk$ may be replaced by a somewhat smaller one.

Given $n$, it thus remains to select $k$. We tentatively  suggest
on intuitive grounds $k =\left\lfloor \frac {n+1}2\right\rfloor$, so that practically
$V$ is twice the left sample median of the $U_i$,
but we have to admit the mathematical arbitrariness
of this choice. We applied this choice in~\cite{Drill.paper} as explained
in Example~\ref{Example.drill} below. This application triggered off  
the present paper.

Before stating our results formally, let us finish this introduction by
stating the obvious unsolved  problems:
Is there any convincing rationale for
choosing $k$ in a particular way? Are there competitive 
summary P-variables which are not functions of  single order statistics?
Are there significantly better summary P-variables  for interesting particular 
cases of our general problem, say in the 
examples of Section~\ref{Sec.Examples} below?

\subsection{Analysis and probability notations}
We assume $n\in\N:=\{1,2,3,\ldots\}$ and put
$\frac 00 := 0$, unless noted otherwise. 
Terms like ``positive'' or  ``increasing'' are used in the wide sense.
We write $\mathrm{Argmax\,} f$ 
for the set of all global maximizers of the real-valued function $f$, 
and $\mathrm{argmax\,} f$ for its unique element if existing. 
We use indicator notation for sets, like $\1_A$, and for  statements,
like $(x\in A) = \1_A(x)$.
Measurable spaces $(\cX,\cA)$ are simply denoted 
by $\cX$ if the $\sigma$-algebra $\cA$ is clear from the context or
irrelevant. $\cB([a,b])$ denotes the 
Borel $\sigma$-algebra of the interval $[a,b]$. 

Let  $\cX$ and $\cY$ be measurable spaces. Then 
$\Prob(\cX)$ denotes the set of all probability measures 
on $\cX$ and $\Mark(\cX, \cY)$ the set of all Markov kernels 
from $\cX$ to $\cY$. Let $P\in\Prob(\cX)$ and $K\in\Mark(\cX, \cY)$. 
Then their product is denoted by $P\otimes K$, that is, 
$(P\otimes K)(C):=\int_\cX \int_\cY
\1_C(x,y)\,K(x,\mathrm{d}y)\,P(\mathrm{d}x)$
for $C$ belonging to the appropriate product $\sigma$-algebra. 
The same symbol $\otimes$ is also used for products of probability measures
and, accordingly, $P^{\otimes n} := \bigotimes_{j=1}^n P $ denotes 
the product of the $n$ identical factors $P\in\Prob(\cX)$.
Thus, e.g.\ in~\eqref{Def.Knkf} in below, an expression like 
$K(x,\cdot)^{\otimes n}$ with $K \in\Mark(\cX, [0,1])$ and $x\in\cX$,
is a product of $n$ probability measures on $[0,1]$, not a product
really involving kernels.

For $P\in\Prob(\cX)$ and $S:\cX\rightarrow\cY$ measurable,
we use the nonstandard notation $S\im P$ for the distribution,
or image measure, of $S$  with respect to $P$. 
We write $\mathrm{B}_{n,p}$ for the binomial distribution 
with counting density  $\mathrm{b}_{n,p}$, that is,
$\mathrm{b}_{n,p}(k)=\binom{n}{k}p^k(1-p)^{n-k}$ for 
$n \in\N\cup\{0\}$, $p\in{[0,1]}$, and $k\in\{0,\ldots,n\}$,
and further $\mathrm{U}_{[a,b]}$ and $\mathrm{U}_{\{1,\ldots,n\}}$
for continuous or discrete uniform distributions on the indicated
sets. Our use of nonitalic letters here allows us to 
use e.g.~$U$ for a random variable.

\subsection{Hypotheses, P-variables, P-kernels,   P-values}
\label{subsec.stats.notation}
In this paper, we call  {\em hypothesis} any set 
$\cP_0\subseteq \Prob(\cX)$ of probability measures 
on the same sample space $\cX$. We might think of $\cP_0$
as a subset of some strictly larger set 
$\cP\subseteq \Prob(\cX) $, with then 
$\cP\setminus \cP_0$  called the alternative,
but only $\cP_0$ matters in the present investigation.

Let $\cP_0\subseteq \Prob(\cX)$ be a hypothesis.
A {\em P-variable} for $\cP_0$
is a statistic  $U :\cX\rightarrow [0,1]$ which 
is under $\cP_0$ stochastically larger than 
the uniform distribution on $[0,1]$, that is,
\[
  P( U  \le \alpha   ) &\le&\alpha 
 \quad\text{ \rm for $P \in\cP_0$ and $\alpha\in[0,1]$}
\]
A {\em P-kernel} for $\cP_0$ is a kernel
$K\in\Mark(\cX,[0,1])$ with 
\la    \label{Q.rand.P-variable}
  \big(P\otimes K \big)  (\cX\times[0,\alpha])  &\le&\alpha 
 \quad\text{ \rm for $P \in\cP_0$ and $\alpha\in[0,1]$}
\al
Finally, a {\em P-value} is any number $\in$ $[0,1]$
resulting from any process we happen to model by a P-variable or 
a P-kernel. 

Clearly, a kernel $K\in\Mark(\cX,[0,1])$  is a P-kernel for $\cP_0$
iff the coordinate projection  $\cX\times [0,1] \ni (x,u) \mapsto u$
is a  P-variable for the hypothesis $\{P \otimes K : P \in \cP_0\}$ on 
the sample space $\cX\times [0,1]$; 
we refer to the latter hypothesis and sample space 
as {\em extended} versions of the original ones.
Conversely, a statistic $U :\cX\rightarrow [0,1]$
is a P-variable for $\cP_0$ iff
the kernel $\cX\times \cB([0,1]) \ni (x,B) \mapsto \delta_{U(x)}(B)$
is a P-kernel for $\cP_0$, and we refer to such a $K$ as {\em deterministic}. 
In statistical practice, non-deterministic P-kernels often arise through 
the application of Monte Carlo tests,
as in Example~\ref{Example.drill} below.

\subsection{The main result}       \label{Sec:Main.result}
Let $n\in\N$ and $k\in \{1,\ldots,n\}$ be fixed in this subsection,
except for the penultimate  sentence.
Let
$U_{k:n}$ denote the $k$-th order statistic on $[0,1]^n$, so
$U_{k:n}(u) := \min\{v\in[0,1] : \sum_{i=1}^n(u_i \le  v) \ge k  \}$
for $u\in[0,1]^n$. We  put
\la
 \psi_{n,k}(p)&:=&p^{-1}\mathrm{B}_{n,p}(\{k,\ldots,n\})\qquad\text{ for }p\in[0,1]
\al
with $\psi_{n,k}(0)$ defined by continuity.
Let $p_{1,1} := 0$. 
If not $n=k=1$, let 
$p_{n,k} :=  \mathrm{argmax\,} \psi_{n,k}$ which exists 
 according to Lemma~\ref{Lem.psi} below.  Finally, let
\la
 c_{n,k} :=\max_{p\in[0,1]}\psi_{n,k}(p) 
\al
and
\la \label{Def.fnk.new}
\qquad f_{n,k}(u) &:=& u\, \psi_{n,k}(p_{n,k}\vee u)
  \,\,\,=\,\,\,\left\{\begin{array}{ll}
  c_{n,k}u  & \text{ for }u\in[0,p_{n,k} ]\\
  \mathrm{B}_{n,u}(\{k,\ldots,n\})   &\text{ for }u\in[p_{n,k},1]
  \end{array}\right.
\al
The following theorem and its  accompanying proposition, 
both proved in Section~\ref{Sec.Aux.Proofs} below, make precise 
the claims from Subsection~\ref{The.problem.etc}
and constitute the main technical result of this paper. 

\begin{thm} \label{Main.result}
The following two conditions on
increasing functions $f:[0,1]\rightarrow [0,1]$ are equivalent:

(i) If  $K$ is a P-kernel for a hypothesis $\cP_0$ 
on a sample space $\cX$, then $ K_{n,k,f} \in \Mark(\cX,[0,1])$
defined  by
\la              \label{Def.Knkf}
  K_{n,k,f}(x,\cdot) &:=& 
    (f\circ U_{k:n}) \im K(x,\cdot)^{\otimes n}
  \qquad \text{ \rm for } x\in\cX
\al
is a P-kernel for $\cP_0$.

(ii) $f\ge f_{n,k}$.
\end{thm}
Since $f_{n,k}:[0,1]\rightarrow [0,1]$ is increasing, 
Theorem~\ref{Main.result} states that $f_{n,k}$ is the smallest,
and thus optimal, eligible 
function satisfying condition~(i).

\begin{prop}\label{fnk-properties}
$f_{n,k}$ is a continuous and increasing bijection of  $[0,1]$
onto $[0,1]$, and is linear with slope $c_{n,k}$ on 
$[0,p_{n,k}]\supseteq  [0,\frac{k-1}{n-1}]$.
We have
\la   \label{cnk.bounds}
  \frac{ \frac nk}{1+5\,k^{-1/3} }     &\le& c_{n,k} \,\,\,\le\,\,\, \frac nk 
\al
and, for $u\in[0,1]$,  
\la \label{fnk.bound.new}
   \frac{1\wedge \frac {nu}k}{1+5\,k^{-1/3} } 
  &\le& f_{n,k}(u)  
 \,\,\,\le\,\,\, 1\wedge(c_{n,k}u)  
 \,\,\,\le\,\,\,1\wedge\frac {nu}k
\al
\end{prop}

In \eqref{cnk.bounds} and \eqref{fnk.bound.new}, the main interest
is in the upper bounds for $c_{n,k}$ and $f_{n,k}$.
The probably improvable lower bounds show  that 
$f_{n,k}(u)$ is asymptotically equivalent to
its upper bound $1\wedge\frac {nu}k$, 
for $k\rightarrow \infty$, 
uniformly in $u$ and in $n\ge k$.   
A numerical computation of $p_{n,k}$, and hence of 
$c_{n,k}$ and $f_{n,k}(u)$, is straightforward due to the monotonicity
properties of $\psi_{n,k}$ given in Lemma~\ref{Lem.psi} below.  

\section{Examples} \label{Sec.Examples}
\subsection{} \label{Example.1}
This is a formalization of the fictitious example from the beginning
of Subsection~\ref{The.problem.etc}. With a notation slightly different
from that of Subsections~\ref{subsec.stats.notation} and~\ref{Sec:Main.result}
(here $\cX^N$ and $\cQ_0$, there $\cX$ and $\cP_0$), let $\cX$ be a sample
space and let $\cP_0\subseteq \Prob(\cX)$. Informally speaking, we want to 
test $\cP_0$ based on several but possibly dependent observables. 
We suppose we know how to handle the i.i.d.~case: For $m\in \N$,
let $U^{(m)}$ be a P-variable for the hypothesis 
$\{P^{\otimes m} : P\in\cP_0 \}$ on the sample space $\cX^m$. 
We further assume that we have $N=\sum_{j=1}^m N_j$
observations $X_{ji}$ coming in $m$ independent groups as
\[
 \big( (X^{}_{11}, \ldots , X^{}_{1N_1}) ,\ldots 
       (X^{}_{m1}, \ldots ,  X^{}_{mN_m}) \big)
\]
with possibly arbitrary dependencies within each group, so that 
our hypothesis can be formalized as 
\[
 \cQ_0 &:=& \left\{ \bigotimes_{j=1}^m Q_j \,:\, Q_j\in \Prob(\cX^{N_j}), 
   \text{ all $N$  $\cX$-marginals equal and $\in$ $\cP_0$}
 \right\}
\]
on the sample space $\cX^N$. Then randomly picking just one observable 
from each of the $m$ blocks and applying $U^{(m)}$ to these yields a valid
P-variable. Formally put,
\[
 K(x,\cdot)&:=&     
\Big( \bigtimes_{j=1}^m \{1,\ldots,N_j\}     \ni i\mapsto
(x_{1i_1},\ldots,x_{mi_m}) \Big) \im \bigotimes_{j=1}^m 
\mathrm{U}_{\{1,\ldots,N_j\}}
\]
for $x= \big( (x^{}_{11}, \ldots, x^{}_{1N_1}) ,\ldots, 
(x^{}_{m1}, \ldots, x^{}_{mN_m}) \big) \in\cX^N$ defines a P-kernel for
$\cQ_0$.  Repeating the random picking $n$ times amounts to considering
the kernel $K_n$ from $\cX^N$ to $[0,1]^n$ defined by 
$K_n(x,\cdot):= K(x,\cdot)^{\otimes n}$, and Theorem~\ref{Main.result}
shows how to transform it into a P-kernel.

\subsection{Primate associations}
\label{Example.drill} In~\cite{Drill.paper} we wanted
to test a hypothesis of ``no association'' between seven
primate species in a certain area (38 km$^2$ within a national 
park in Cameroon). The data, obtained by patrolling
3284 km over 217 days and recording the species composition of 
612 observed primate clusters, is a matrix
\[
  x &=& (x_{ij}) \,\,\in\,\,\{0,1\}^{612\times 7} \,\,\,=:\,\,\cX
\]
with
\[
 x_{ij} &:=& \left\{\begin{array}{ll}1\\ 0\end{array} \right\}
 \,\text{ if species $j$ is }\,
 \left\{\begin{array}{ll}\text{present} \\ \text{absent}\end{array} \right\}
 \text{ in cluster }i
\]

For example, the clusters observed on the first two and last two  observation
days  yielded
{\footnotesize

\begin{center}
\begin{minipage}{.9\linewidth}
\begin{Schunk}
\begin{Soutput}
          date putty mona redear crowned drill mangabey redcol
1   2006-02-03     1    0      1       0     0        0      0
2   2006-02-03     1    0      0       0     0        0      0
3   2006-02-03     0    0      1       1     0        0      1
4   2006-02-04     0    0      0       1     0        0      0
5   2006-02-04     1    0      1       0     0        0      0
6   2006-02-04     1    0      1       0     0        0      0
7   2006-02-04     0    0      0       0     0        1      0
8   2006-02-04     1    0      0       0     0        0      0
9   2006-02-04     1    0      0       1     0        0      1
607 2008-01-17     1    0      1       0     0        1      0
608 2008-01-17     1    0      0       0     0        0      0
609 2008-01-17     1    0      1       1     0        0      1
610 2008-01-17     1    1      0       0     0        0      0
611 2008-01-18     0    1      0       1     0        0      0
612 2008-01-18     1    0      0       1     0        0      0
\end{Soutput}
\end{Schunk}
\end{minipage}
\end{center}}
\noindent Here the columnnames starting from ``putty''  are abbreviations for English
species names, with  the corresponding scientific names obtainable 
from~\cite[p.~130, Fig.~2]{Drill.paper}. For example, ``putty'' is 
an abbreviation for ``putty-nosed guenon'' or ``\textit{Cercopithecus nictitans}''.
    
We write  $x_{\cdot+}$ and  $x_{+\cdot}$ for
the  row and column sums vectors  of $x\in\cX$. 
The ``no association'' hypothesis is formalized as
\[
  \cP_0 &:=& \big\{P\in\Prob( \cX) : P(\{x\})=P(\{y\})\text{ if }
          x_{\cdot+} = y_{\cdot+}  \text{ and }x_{+\cdot}= y_{+\cdot}\big\}
\]
This formalization is common in ecology and in social network analysis,
see~\cite{Adlerrochen} for some references,
it is certainly debatable,  but it is for instance  larger,
and hence more interesting to reject,
than  assuming that rows are independent and each represents a 
random sample from the species' which is simple  
(in the usual sense of all samples being equally likely)
if conditioned on its possibly random size.  
A submodel of $\cP_0$, consisting of all laws of 
$\cX$-valued random variables $X=(X_{ij})$ with  independent
indicators $X_{ij}$  with success
probabilities $\alpha_i\beta_j/(1+\alpha_i\beta_j)$
for some $\alpha_i,\beta_j\in{[0,\infty[}$,
was originally proposed by Rasch for a completely different 
situation, see~\cite[page 75]{Rasch}.

Let $T:\cX \rightarrow \R$ be a statistic to be used for 
testing $\cP_0$, with sufficiently large values to be regarded as significant.
The corresponding P-variable $U$ appears impossible to compute, but 
Besag and Clifford  constructed in~\cite{Besag.Clifford.1989}
a P-kernel for  $\cP_0$ using Markov chains 
as follows:

Fix a length $N\in\N$. Given the observation $x\in\cX$, let 
$\tau$ be uniformly distributed  on $\{1,\ldots,N\}$, then let 
$X_\tau := x$, then construct $X_{\tau+1},\ldots,X_N$ as a Markov chain 
starting from $X_\tau $ and using stationary transition probabilities 
with the uniform distribution on 
$\cX(x):=\{y\in\cX: y_{\cdot+} = x_{\cdot+}\text{ and }y_{+\cdot}=x_{+\cdot}\}$
as their stationary distribution, then analogously construct
$X_{\tau-1},\ldots, X_1$ using the corresponding reversed transition
probabilities. Then take the observed value of 
\[ \frac 1N \sharp\{ t : T(X_\tau) \le T(X_t) \}\] 
as our P-value. This is modelled by a valid 
P-variable, on a suitably extended sample space,  since $\tau$ and the 
chain $(X_1,\ldots,X_N)$ are independent under  the hypothesis $\cP_0$.
See~\cite{Besag.Clifford.1989} and~\cite{Adlerrochen} 
for more details.

With a certain $T$ and with $N= 10^8$, the above 
yielded a P-value of $0.006$ for our data.
(This may  be regarded as a P-value obtained with a 
practically, though not theoretically, 
deterministic P-kernel, as running a few repetitions
of the described Besag-Clifford procedure showed.) 
But perhaps the null model we then wanted  to reject 
is inappropriate not due to an interesting
scientific reason but, for example,  due to a possible reobservation
of the same primate cluster on the same day. Hence we 
reanalyzed the data by 
randomly picking  just one cluster per day. Doing this $n=1000 $ times 
yielded P-values with sample quartiles $0.01$, $0.03$, $0.09$, maximum
$0.7$, and the following histogram:

\begin{center}
%
%
\vspace{-5em}
\includegraphics[width=\linewidth,
  keepaspectratio
  ]{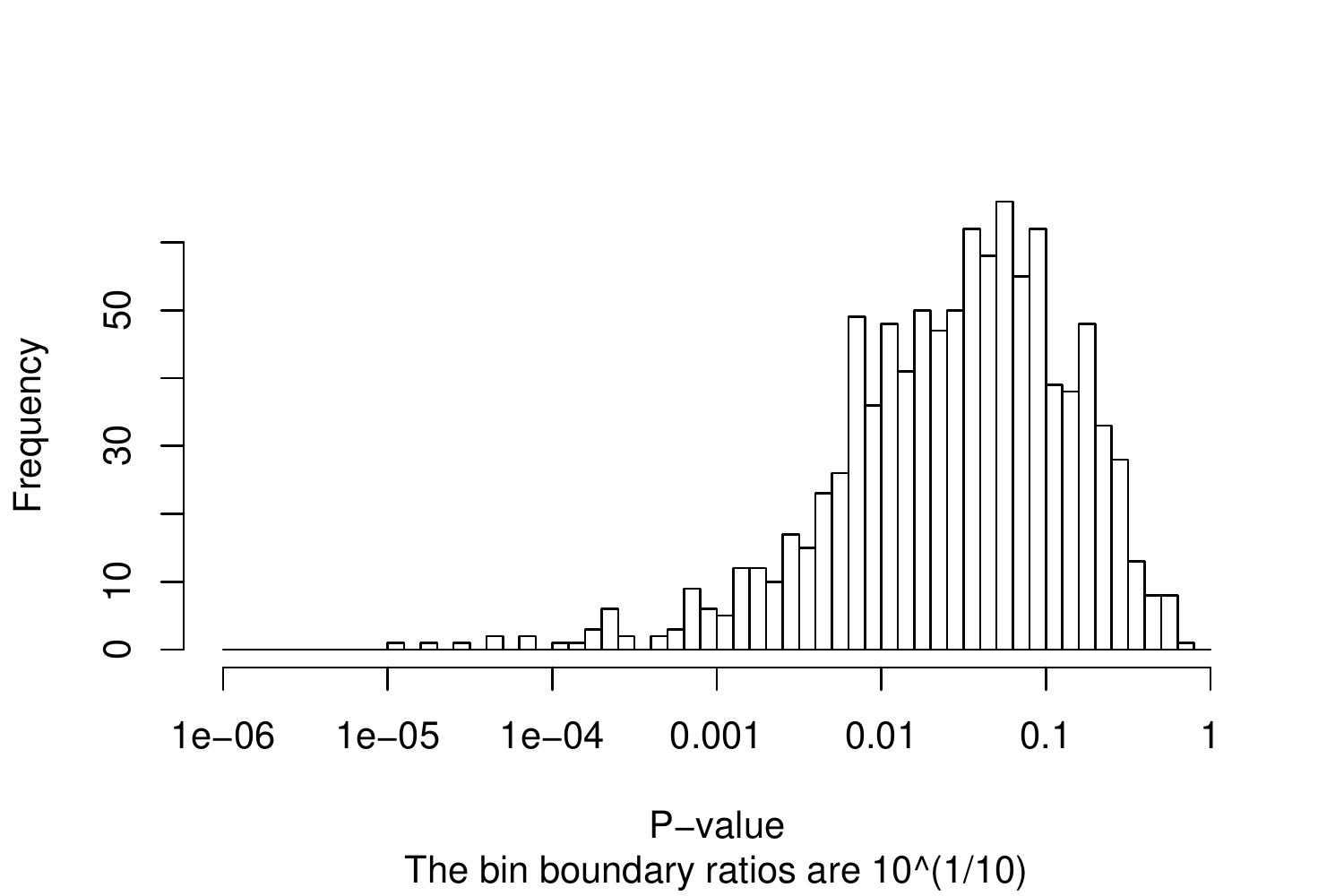}%
\end{center}
Now what to report as an appropriate summary P-value? 
With the mathematical formalization analogous to Example~\ref{Example.1},
we decided to use Theorem~\ref{Main.result} with $k=n/2$, yielding $c_{n,k} =1.846$,
and we thus reported   $0.03\times1.846 = 0.06$ 
in  \cite[Table II, last column]{Drill.paper}.

\subsection{Testing Hardy-Weinberg equilibrium using complex survey data} 
In~\cite{Li.Graubard.2009}, Li and Graubard
propose certain tests for  Hardy-Weinberg equilibrium and related hypotheses 
for humans,  based on data obtained with certain  complex survey designs,
often leading to samples with indivuals belonging to the same 
household. To address the dependence problem due to people
living in the same household possibly being   blood related,
they ``selected one person randomly from each household
to form a subsample''. Their  ``analyses using 10 subsamples 
produce(d) consistent results \ldots and in general agreed
with results using the full sample analyses''.   
See \cite[pages 1103, 1101]{Li.Graubard.2009}. 
So here the P-variable combination problem
is not properly addressed, at least not explicitly.
Perhaps one could interpret ``produce(d) consistent results'' 
as looking for the maximum of the P-variables and multiplying
it by the then trivial correction constant $c_{10,10}=1$, 
but such a procedure should  have been fixed 
before looking at the 10 actual P-values.

\section{Auxiliary results and  proofs}\label{Sec.Aux.Proofs}
As in Subsection \ref{Sec:Main.result}, let $n\in\N$ and $k\in \{1,\ldots,n\}$
be fixed. In Lemma~\ref{lem.binom} and in the proofs of Lemma~\ref{Lem.psi}
and Proposition~\ref{fnk-properties},
we write 
\[
 \phi(p)&:=&\phi_{n,k}(p)\,\,\,:=\,\,\, \mathrm{B}_{n,p}(\{k,\ldots,n\}) 
\]
for $p\in[0,1]$, so that $\psi_{n,k}(p)= p^{-1}\phi(p)$,
with the right hand side  defined by continuity for $p=0$.
We recall that $\frac{k-1}{n-1}:=0$ in case of $k=n=1$. 
 
\begin{lem}\label{lem.binom}
The function $\phi$ is continuous, strictly increasing, strictly convex 
on $[0,\frac{k-1}{n-1}]$,  concave on $[\frac{k-1}{n-1},1]$,
and satisfies $\phi(0)=0$.
\end{lem}
\begin{proof}For $p\in{]0,1[}$, we have 
$\phi'(p)=n\,\mathrm{b}_{n-1,p}(k-1)$ and  
\la     \label{phi.2.Strich}
 \phi''(p) &=& \frac{n\,\mathrm{b}_{n-1,p}(k-1)}{p\,(1-p)}
 \big(k-1-(n-1)p\big)
\al
and everything is obvious.
\end{proof}

\begin{lem}\label{Lem.psi}
We have 
$\psi_{1,1}(p)=1$ for $p\in[0,1]$, while for $n\ge 2$,
$p_{n,k} = \mathrm{argmax\,} \psi_{n,k}$ exists and belongs to $[\frac{k-1}{n-1},1]$,
with $\psi_{n,k}$ strictly increasing on $[0,p_{n,k}]$
and strictly decreasing on $[p_{n,k},1]$. We further have
\la     \label{Alt.rep.fnk}
    f_{n,k}(u) &=& u \max_{p\in[u,1]} \psi_{n,k}(p)\qquad\text{ for }u\in[0,1]
\al
\end{lem}
\begin{proof}
The case of $n=1$ is trivial, so let us assume $n\ge2$.
Let us abbreviate $\psi:=\psi_{n,k}$.
On  $[0,\frac{k-1}{n-1}]$, 
the function $\phi$ is strictly convex by Lemma~\ref{lem.binom},
hence $\psi(p) = \frac{\phi(p)-\phi(0)}{p-0}$ is strictly increasing there.
For $p\in{]0,1]}$, we have 
\[
 \psi'(p) &=& p^{-2}\,\big(p\,\phi'(p) -\phi(p) \big) 
 \,\,\,=:\,\,\, p^{-2}\,\omega(p)\\
 \omega'(p) &=& p\,\phi''(p) \,\,\,=\,\,\,
 \frac{n\,\mathrm{b}_{n-1,p}(k-1)}{1-p} \big(k-1-(n-1)p\big)
\] 
by~\eqref{phi.2.Strich}.
So $\omega'(p) <0$ for $p\in{]\frac{k-1}{n-1}, 1[}$.
Hence  there is at most one $p_0\in{]\frac{k-1}{n-1}, 1[}$
with $\psi'(p_0)=0$, and if such a $p_0$ exists, then 
$\psi'$ changes sign from plus to minus at $p_0$.
Thus  $p_{n,k}=p_0$  if $p_0$ as above exists, and 
otherwise $\psi$ is strictly monotone 
on $\left[\frac{k-1}{n-1},1\right]$
and 
$p_{n,k}$ is the appropriate element of $\{\frac{k-1}{n-1}, 1\}$.
This proves all claims about $\psi$, and~\eqref{Alt.rep.fnk}
follows, see the by now justified definition in~\eqref{Def.fnk.new}.
\end{proof}

\begin{lem} We have
\la
 \sup_{p\,\in\,{\left]0,1\right]}}\, \frac{k}{np}\mathrm{B}_{n,p}(\{k,\ldots,n\})     
  &\le&  1   \label{Bin.Markov}     \\
 \max_{p\,\in\,\left[\frac kn,1\right]}
 \,\frac{k}{np}\mathrm{B}_{n,p}(\{k,\ldots,n\}) 
  &\ge&  \left(1+ 5\,k^{-1/3}\right)^{-1}   \label{One-sided.Cheb.Markov}  
\al
\end{lem}
\begin{proof} Relation \eqref{Bin.Markov}, 
following from 
$k\,\mathrm{B}_{n,p}(\{k,\ldots,n\})\le\sum_{j=0}^n j\,\mathrm{b}_{n,p}(j)$,
is a simple Markov inequality. 
To prove \eqref{One-sided.Cheb.Markov},
let $g(p) := \frac{k}{np}\mathrm{B}_{n,p}(\{k,\ldots,n\}) $
for $p\in  \left[\frac kn,1\right]$.
If also $p\neq \frac kn$, then 
$k < \mu(\mathrm{B}_{n,p})$
and the one-sided Chebyshev inequality 
\cite[page~476, (3.2)]{Karlin.Studden} yields 
\[
 \mathrm{B}_{n,p}(\{k,\ldots,n\}) &\ge& 1- \mathrm{B}_{n,p}(\{0,\ldots,k\})
 \,\,\,\ge \,\,\, \frac{(np-k)^2}{np(1-p) + (np-k)^2} 
\]
and hence
\[
 g(p)  &\ge& 
 \left( \frac{np}k +
   \frac{p^2(1-p)}{k\left(p-\frac{k}{n}\right)^2}\right)^{-1}
 \,\,\,\ge\,\,\,
 \left( \frac{np}k +
   \frac{p^2}{k\left(p-\frac{k}{n}\right)^2}\right)^{-1} 
   \,\,\,=:\,\,\, h(p)
\]
If now in particular $p := (k+k^{2/3})/n$ satisfies $p\le 1$,
then we get 
\[
 h(p)  &=& \left(1+k^{-1/3} + k^{-1/3}\,( 1+k^{-1/3})^2 \right)^{-1}
   \,\,\,\ge \,\,\, (1+5\,k^{-1/3})^{-1}
\]
and otherwise we have 
$\frac kn > (1+k^{-1/3})^{-1}_{}$
and hence 
$g(1)= \frac kn > (1+5\,k^{-1/3})^{-1}_{}$. 
Thus \eqref{One-sided.Cheb.Markov} holds in every case.
\end{proof}

\begin{proof}[Proof of Proposition~\ref{fnk-properties}]
Let us abbreviate $f:=f_{n,k}$.
By the definition in~\eqref{Def.fnk.new}, $f$ is positive
and continuous, since $\psi_{n,k}$ is so, and also $f(0)=0$.
By the second representation in~\eqref{Def.fnk.new},
$f$ is strictly increasing, using $c_{n,k}>0$ and $\phi$
strictly increasing, and we have $f(1)=1$.
With  $p_{n,k} \ge \frac{k-1}{n-1}$ from Lemma~\ref{Lem.psi},
this proves the claimed mapping properties of $f$.

Inequalities~\eqref{cnk.bounds} follow from~\eqref{Bin.Markov} 
and~\eqref{One-sided.Cheb.Markov}.
For $u\le \frac{k}{n}$, inequality \eqref{One-sided.Cheb.Markov} 
also yields the first inequality in~\eqref{fnk.bound.new}, since 
then~\eqref{Alt.rep.fnk} shows that  $f(u)$ is at least
$\frac {nu}k$ times the left hand side of~\eqref{One-sided.Cheb.Markov},
while for $u\ge \frac{k}{n}$ the inequality to be proved persists, since
the left hand side is constant and the right hand side increasing.
Next, $f(u)\le c_{n,k}u$  for $u\in[0,1]$
follows from $\psi_{n,k}(p_{n,k}\vee u) \le c_{n,k}$,
and $f(u)\le 1$ was proved above.
The final inequality in~\eqref{fnk.bound.new} follows
from~\eqref{cnk.bounds}.
\end{proof}

Now let us recall the definition and a well-known lemma concerning 
the convex order: 
For $P,Q\in\Prob(\R)$ with finite means $\mu(P)$ and $\mu(Q)$,
one writes $P\cxle Q$ if $\int \phi \dd P \le \int \phi \dd Q$  for every convex
function $\phi:\R \rightarrow {]{-\infty},\infty]}$. Then necessarily $\mu(P)=\mu(Q)$.

\begin{lem}\label{lem.convex.order}
Let $a,b\in\R$ with $a<b$ and let $P\in\Prob([a,b])$ with mean $\mu$.
Then $\delta_\mu \cxle P \cxle \frac{b-\mu}{b-a}\delta_a 
+ \frac{\mu-a}{b-a}\delta_b$.  
\end{lem}
\begin{proof}This is trivial if $\mu\in\{a,b\}$. Otherwise,  
for $\phi$ convex on $[a,b]$, apply $\int  \ldots \dd P(t)$
to $\phi(\mu)+\phi'(\mu+)(t-\mu)\le\phi(t)\le\frac{b-t}{b-a}\phi(a)
+\frac{t-a}{b-a}\phi(b)$. 
\end{proof}

\begin{lem}\label{lem.Argmax}
Let $\alpha,c\in[0,1]$ and let $\phi:[0,1]\rightarrow\R$
be a continuous and increasing function, convex on $[0,c]$ and concave 
on $[c,1]$.  Then the set
\la      \label{Argmax.phi.alpha}
 \mathrm{Argmax\,} \Big(\int_{[0,1]} \phi\dd P :   
 P\in\Prob([0,1])\text{ \rm with }\mu(P)\le\alpha   \Big)
\al
has nonempty intersection with 
\la\label{Argmax.phi.alpha.cand}
 \left\{ \big(1-\frac{\alpha}t\big)\delta_0 +\frac{\alpha}t\delta_t 
  \,:\, t\in[\alpha\vee c,1] \right\}
\al
\end{lem}
\begin{proof}Let $\cP$ denote the set in \eqref{Argmax.phi.alpha}.
As $\{P\in\Prob([0,1]): \mu(P)\le \alpha\}$ is nonempty and closed 
with respect to convergence in distribution, the
compactness of $[0,1]$ and the continuity of $\phi$
imply by Prohorov's theorem (in the elementary Helly case, 
see e.g.~\cite[Section 25]{Billingsley})
that $\cP$ is nonempty. So let $P\in\cP$. If 
$\mu(P)< \alpha$, then $P_s := (1-s)P+s\delta_1\in\Prob([0,1])$ 
with $s:=(\alpha-\mu(P))/(1-\mu(P))$ satisfies
$\mu(P_s)=\alpha$ and, since $\phi$ is increasing, 
$\int \phi \dd P \le \int \phi\dd P_s$. 
Hence we may assume $\mu(P)=\alpha$ in what follows.    
If $P([0,c])>0$ and $c>0$, then 
we apply the second inequality in Lemma~\ref{lem.convex.order}
with  $a:=0$ and $b:=c$ to  the law 
$\cB([0,c]) \ni B \mapsto P(B)/P([0,c])$  with  mean $\lambda$, say, 
to see that $\int\phi \dd P \le \int\phi \dd Q$
where $Q:= P([0,c])(\frac{c-\lambda}{c}\delta_0 + \frac{\lambda}{c}\delta_c)
+\1_{]c,1]}P \in \cP$ 
also satisfies $\mu(Q)=\alpha$.    
Hence we may also assume $P({]0,c[})=0$ in what follows. 
If, finally, $P([c,1])>0$ and
$c<1$, then we apply the first inequality in Lemma~\ref{lem.convex.order}
with $a:= c$ and $b:= 1$ to the law 
$\cB([c,0]) \ni B \mapsto P(B)/P([c,1])$ with  mean $\rho$, say,
to   see that $\int\phi \dd P \le \int\phi \dd R$
where $R := P(\{0\})\delta_0 + P([c,1])\delta_\rho$ if $c>0$,
and $R := \delta_\rho$ if $c= 0$, so that $R$   belongs to $\cP$ and to
the set in   \eqref{Argmax.phi.alpha.cand}.
\end{proof}

From Lemmas~\ref{lem.Argmax} und~\ref{lem.binom}, 
recalling~\eqref{Alt.rep.fnk} and observing that each member of the 
set in~\eqref{Argmax.phi.alpha.cand} has mean $\alpha$, we get 
\begin{lem}\label{lem.main} For every $\alpha\in[0,1]$,  we have  
\[
\max \left\{\int \mathrm{B}_{n,t}(\{k,\ldots,n\})\dd P(t)\,:\,
   P\in\Prob([0,1]), \mu(P)\le\alpha    \right\}
 &=& f_{n,k}(\alpha)
\]
\end{lem}

\textsc{Proof of Theorem \ref{Main.result}.}
Let $f:[0,1]\rightarrow[0,1]$ be an increasing function, 
fixed in the entire proof.
If  $\alpha\in[0,1]$  and  
\la  \label{Eq:Def.I.beta}
 I \,:=\,f^{-1}([0,\alpha]),\quad \beta:=\sup I,\quad \sup\emptyset :=0
\al
so that  $I=[0,\beta[$ or $I=[0,\beta]$, 
and if further $K$ is a P-kernel for a hypothesis $\cP_0$ 
on a sample space $\cX$, $ K_{n,k,f}$ is defined  as
in~\eqref{Def.Knkf},  $P\in\cP_0$  and 
\la      \label{Def.P.tilde}
  \tilde{P} &:=& K(\cdot,I) \im P      \,\,\in\,\, \Prob([0,1])
\al
then we have 
\la
 \nonumber \big(P \otimes K_{n,k,f}\big)(\cX\times[0,\alpha])  
 &=& \int_\cX \big( U_{k:n}\im K(x,\cdot)^{\otimes
   n}_{}\big)(I )   \, P(\mathrm{d}x) \\
\label{Eq:uses.orderstats.cdf} &=&\int_\cX \mathrm{B}_{n,K(x,I )}(\{k,\ldots,n\}) 
 \, P(\mathrm{d}x) \\
\label{16.5} &=&\int_{[0,1]} \mathrm{B}_{n,t}(\{k,\ldots,n\}) \,\tilde{P}(\mathrm{d}t)\\
\label{le.fnk} &\le& f_{n,k}(\beta) 
\al
where at~\eqref{Eq:uses.orderstats.cdf}
we have used the standard formula for the distribution functions of
order statistics \cite[Exercise 14.7, unfortunately missing in 
the third edition from 1995]{Billingsley},
and where inequality~\eqref{le.fnk} follows from Lemma~\ref{lem.main}
applied to $\beta$ and $\tilde{P}$,  as 
$\mu( \tilde{P})=\int_\cX K(\cdot,I)\dd P
= \big(P\otimes K\big)(\cX\times I)
\le \big(P\otimes K\big)(\cX\times[0,\beta])
 \le \beta$
by the P-kernel assumption~\eqref{Q.rand.P-variable}. 

Now let us assume condition (ii). Then,  for every $\alpha\in[0,1]$,
the corresponding $\beta$ from~\eqref{Eq:Def.I.beta} satisfies
$\beta \le \sup f_{n,k}^{-1}([0,\alpha]) = f_{n,k}^{-1}(\alpha)$,
using Proposition~\ref{fnk-properties} for the last equality, and hence 
\[
 f_{n,k}(\beta) &\le& f^{}_{n,k}(f_{n,k}^{-1}(\alpha)) \,\,\,=\,\,\,\alpha  
\]
so that $ \big(P \otimes K_{n,k,f}\big)(\cX\times[0,\alpha]) \le \alpha$.
Hence (i) follows.

To prepare for the proof of the converse implication, 
let us consider  $\cX:=[0,1]$,  $\cP_0:=\{P\}$ with $P :=\mathrm{U}_{[0,1]}$, 
and $t\in[0,1]$. Then 
\[
 K(x,\cdot) &:=& t \delta_{xt} + (1-t)\mathrm{U}_{[t,1]} 
 \qquad\text{ for }x\in\cX
\]
defines a P-kernel $K\in\Mark(\cX,[0,1])$, since for $\alpha\in[0,1]$
we have 
\la
 K(x,[0,\alpha]) &=& t\cdot(xt\le \alpha) + (\alpha-t)_+ \label{Kxalpha1} \\
 K(x,[0,\alpha[) &=& t\cdot(xt< \alpha) + (\alpha-t)_+\label{Kxalpha2}
\al
for $x\in\cX$ and hence, using just \eqref{Kxalpha1}, 
\[
 \big(P\otimes K\big)(\cX\times [0,\alpha]) 
  &=& \alpha\wedge t + (\alpha-t)_+ \,\,\,=\,\,\, \alpha
\] 
The identities \eqref{Kxalpha1} and \eqref{Kxalpha2} further yield
\la   \label{Im.measure}
 K(\cdot,[0,\alpha])\im P &=&     K(\cdot,[0,\alpha[)\im P \\
 &=& \big(1-\frac{\alpha}{t} \wedge 1\big)\delta_{(\alpha-t)_+}
  +  \big(\frac{\alpha}{t} \wedge 1\big) \delta_{t+(\alpha-t)_+} \nonumber
\al
for $\alpha \in[0,1]$. Let us now fix $\alpha\in[0,1]$
for the rest of this paragraph, define $I$ and $\beta$ as in~\eqref{Eq:Def.I.beta}
above, 
recall $\psi_{n,k}$ and $p_{n,k}$ from Section~\ref{Sec:Main.result},
and specialize to $t:= p_{n,k}\vee \beta$. 
Then~\eqref{Def.P.tilde} and~\eqref{Im.measure}, the latter with 
$\beta$ in place of $\alpha$, yield
\[ 
 \tilde{P} &=&  \big(\,1-\frac \beta t\,\big)\delta_0 + \frac{\beta}{t}\delta_t
\]
and the integral in~\eqref{16.5} reduces to 
$\beta\, \psi_{n,k}(t) = f_{n,k}(\beta)$ by the definition
in~\eqref{Def.fnk.new}, and hence we get
\la    \label{Eq.in.le.fnk}
     \big(P \otimes K_{n,k,f}\big)(\cX\times[0,\alpha]) &=& f_{n,k}(\beta)
\al

Now let us assume condition~(i). Then the left hand side  of~\eqref{Eq.in.le.fnk}
is at most $\alpha$, and hence we get 
\la  \label{Ineq.17} 
  f_{n,k}^{}\big(\sup f^{-1}_{}([0,\alpha]) \big) 
  &\le& \alpha\qquad\text{ for }\alpha\in[0,1]
\al
We recall Proposition~\ref{fnk-properties} for the properties of $f_{n,k}$ 
used below. If now $f(\alpha_0) <  f_{n,k}^{}(\alpha_0)$ for some $\alpha_0 \in[0,1]$,
then $\alpha_0>0$ as $ f_{n,k}^{}(0)=0$, and hence the continuity 
of $f_{n,k}^{}$ yields an $\alpha_1\in[0,\alpha_0[$ with $f(\alpha_0) \le
f_{n,k}^{}(\alpha_1)$, and  thus 
\[
 \sup f^{-1}_{}([0, f_{n,k}^{}(\alpha_1)]) 
\,\,\,\ge \,\,\,\sup f^{-1}_{}([0, f(\alpha_0)]) 
&\ge& \alpha_0 \,\,\,> \,\,\, \alpha_1   
\]  
and hence  
\[
 f_{n,k}^{}\big(\sup f^{-1}_{}([0, f_{n,k}^{}(\alpha_1)])\big) 
 &>&  f_{n,k}^{}( \alpha_1)   
\]
in contradiction to~\eqref{Ineq.17} with $\alpha:=f_{n,k}(\alpha_1)$.  Thus (ii) holds.
\hfill \Halmos

\section*{Acknowledgements}
Thanks are due to  my coauthors from \cite{Drill.paper} for 
the cooperation leading to the problem treated in the present
paper, in particular to  Christos Astaras for
collecting the data and initiating our joint work, and to 
Stefan Krause for numerous discussions about the data analysis,
and also for actually doing the computations mentioned in 
Example~\ref{Example.drill}.
I further thank Jannis Dimitriadis for proofreading and 
Todor Dinev for help with Sweave.

\end{document}